\documentclass[letterpaper, 10 pt, conference]{ieeeconf}  

\IEEEoverridecommandlockouts                              

\usepackage{hyperref}	
\usepackage{xcolor}		

\usepackage{amsfonts}   
\usepackage{amsmath}    
\usepackage{amssymb}
\usepackage{mathtools}
\usepackage{bm}         
\usepackage{algorithm}	
\usepackage{algpseudocode}
\usepackage{svg}	
\usepackage{subfig}	
\usepackage{setspace}
\usepackage{amsthm}     

\newtheorem{proposition}{Proposition}

\theoremstyle{definition}

\theoremstyle{remark}

\newtheorem{remark}{Remark}
\newtheorem{assumption}{Assumption}

\usepackage{tikz}
\usetikzlibrary{spy}
\usepackage[utf8]{inputenc}
\usepackage{pgfplots}
\usepgfplotslibrary{groupplots,dateplot}
\usetikzlibrary{patterns,shapes.arrows}
\pgfplotsset{compat=newest}
\def\axisdefaultheight{110pt}

\DeclareMathOperator*{\argmin}{arg\,min}

\DeclareRobustCommand{\vect}[1]{\bm{#1}}
\title{\LARGE \bf
Learning-Based Model Predictive Control for Piecewise Affine Systems with Feasibility Guarantees
}

\author{Samuel Mallick, Azita Dabiri, Bart De Schutter
\thanks{
This project has received funding from the European Research Council (ERC) under the European Union's Horizon 2020 research and innovation programme (Grant agreement No. 101018826 - ERC Advanced Grant CLariNet).}
\thanks{
The authors are with the Delft Center for Systems and Control, Delft University of Technology,
Delft, The Netherlands,
{\tt\small \{s.h.mallick, a.dabiri, b.deschutter\}@tudelft.nl}}%
}

\begin{document}

\maketitle
\thispagestyle{empty}
\pagestyle{empty}

\begin{abstract}
Online model predictive control (MPC) for piecewise affine (PWA) systems requires the online solution to an optimization problem that implicitly optimizes over the switching sequence of PWA regions, for which the computational burden can be prohibitive.
Alternatively, the computation can be moved offline using explicit MPC; however, the online memory requirements and the offline computation can then become excessive.
In this work we propose a solution in between online and explicit MPC, addressing the above issues by partially dividing the computation between online and offline.
To solve the underlying MPC problem, a policy, learned offline, specifies the sequence of PWA regions that the dynamics must follow, thus reducing the complexity of the remaining optimization problem that solves over only the continuous states and control inputs.
We provide a condition, verifiable during learning, that guarantees feasibility of the learned policy's output, such that an optimal continuous control input can always be found online.
Furthermore, a method for iteratively generating training data offline allows the feasible policy to be learned efficiently, reducing the offline computational burden.
A numerical experiment demonstrates the effectiveness of the method compared to both online and explicit MPC. 
\end{abstract}


\section{INTRODUCTION}
Traditional online model predictive control (MPC) for piecewise affine (PWA) systems requires the online solution to an optimization problem that optimizes over, not only the continuous states and control actions, but also implicitly over the sequence of PWA regions.
Through the introduction of auxiliary variables this problem can be reformulated as a mixed-integer linear/quadratic program (MILP/MIQP) \cite{bemporadControlSystemsIntegrating1999}.
Unfortunately, MILP/MIQPs are NP-hard problems, and in the context of MPC the computational burden can grow exponentially with the size of the system and the length of the prediction horizon \cite{wolsey2014integer}, often limiting a real-time implementation.

Explicit MPC methodology has been proposed to address this issue, shifting the computational burden offline \cite{oberdieckExplicitHybridModelPredictive2015}.
Here, a multiparametric MILP/MIQP (mp-MILP/MIQP) is solved offline, characterizing the optimal control input as a function of the system state.
This solution takes the form of a partition of the state space, with a PWA control law associated to each region of the partition, such that the online computation reduces to a lookup table and a function evaluation.
However, the computational complexity of the mp-MILP/MIQP is such that this approach is in general limited to small-scale systems.
Furthermore, the solution's partition can be exceedingly complex, requiring significant online memory space to store the regions.

As an alternative to offline MPC, machine learning methods have been employed to address the computational burden of mixed-integer-based MPC.
In \cite{mastiLearningBinaryWarm2019} a learned policy provides warm-start guesses for an MIQP solver, in order to decrease the required solution time; however, in the worst case the computational burden remains exponential in the problem size.
In \cite{mastiLearningApproximateSemiExplicit2020, cauligi2021coco, cauligiPRISMRecurrentNeural2022, da2024integrating} a policy is learned that selects and fixes the configuration of integer variables prior to solving the MPC optimization problem, which is then convex, and far less demanding to solve.
However, these approaches do not guarantee that the remaining optimization problem is feasible after the integer variables are chosen.

In light of the above issues, this work presents a learning-based MPC controller for PWA systems in which the feasibility of the learned policy's output is guaranteed.
We propose an approach that follows the principle of \cite{mastiLearningApproximateSemiExplicit2020, cauligi2021coco, cauligiPRISMRecurrentNeural2022, da2024integrating} in decoupling the selection of discrete components, i.e., the sequence of PWA regions, and the continuous components, i.e., the state and control input trajectories, into a learned policy and a linear MPC optimization problem, respectively.
Leveraging the properties of PWA dynamics, we formulate a class of classifiers whose output can be verified for feasibility during learning.
Furthermore, we introduce an iterative procedure for generating training data such that the feasible policy is learned efficiently.

The contributions of this work are as follows. 
As only a linear MPC problem is solved online, relative to the original online MPC problem the proposed approach has a much lower online computational burden.
Furthermore, the online memory requirement for the learned policy is significantly less than that of offline MPC, as learning the sequence of PWA regions is less complex than learning the optimal control signal.
Finally, in contrast to existing approaches, the learned policy provably selects a feasible output.

The remainder of this paper is organized as follows. 
In Section \ref{sec:problem_setting} the problem setting is defined.
In Section \ref{sec:properties} useful structural properties of the MPC optimization problem are provided.
Section \ref{sec:learning_based_mpc} introduces our proposed controller, which is then demonstrated in a numerical example in Section \ref{sec:example}.
Finally, Section \ref{sec:conclusions} concludes the paper. 

\subsection{Definitions}
\begin{color}{black}
We define a \emph{polytope} as the intersection of a finite number of (open or closed) halfspaces, which is then convex.
A collection of sets $P_1,\dots,P_l$ forms a \emph{partition} of the set $P$ if $\bigcup_{i = 1}^l P_i = P$, $P_i \neq \emptyset$ for all $i$, and $P_i \cap P_j = \emptyset$ for all $i \neq j$.
If each set $P_i$ is a polytope they form a \emph{convex partition} of $P$.
For a set $\mathcal{X}$ we denote its convex hull as $\text{Conv}(\mathcal{X})$, and its closure as $\bar{\mathcal{X}}$.
The set of vertices of a polytope $P$ is denoted as $\text{Vert}(P)$.
Finally, indexing via subscripts represents closed-loop time steps, e.g., $x_k$, while brackets represent time steps within an MPC prediction horizon, e.g., $x(k)$.\end{color}

\section{PROBLEM SETTING}\label{sec:problem_setting}
\begin{color}{black}
Consider the discrete-time PWA system with state $x \in \mathcal{X} \subseteq \mathbb{R}^n$, input $u \in \mathcal{U} \subseteq \mathbb{R}^m$, and affine dynamics over a convex partition $\{P_i\}_{i=1}^l$ of $\mathcal{X}$ with $l$ polytopes indexed by $i \in \mathcal{L}=\{1, \dots, l\}$, referred to as regions,
\begin{equation}\label{eq:dynamics}
	x(k+1) = A_i x(k) + B u(k) + c_i, \: x(k) \in P_i.
\end{equation}
The states and inputs are constrained to the sets $\mathcal{X}$ and $\mathcal{U}$, respectively.
\begin{assumption}
	The dynamics are assumed to be continuous, i.e., $A_ix + Bu + c_i = A_j x + Bu + c_j$ for all $x \in \bar{P}_i \cap \bar{P}_j$, such that the dynamics can be considered over the closures $\{\bar{P}_i\}_{i \in \mathcal{L}}$ without the need for a set-valued function.
	The sets $\mathcal{X}$ and $\mathcal{U}$ are assumed to be compact polytopes with the origin in their interior.
	There exists a terminal set $\mathcal{X}_\text{f} \subseteq \mathcal{X}$ that is a polytope, assumed to contain the origin in its interior, and to be positive invariant under some linear control laws, i.e., there exists $K_i \in \mathcal{R}^{m \times n}$, for all $i$ such that $P_i \cap \mathcal{X}_\text{f} \neq \emptyset$, such that $(A_i + B K_i)x \in \mathcal{X}_\text{f}$ for all $x \in P_i \cap \mathcal{X}_\text{f}$.
\end{assumption}

We consider an MPC controller for the system \eqref{eq:dynamics}.
Define $\bm{\delta} = \big(\delta(0), \dots, \delta(N)\big) \in \mathbb{Z}_{[0, l]}^{N+1}$ as a switching sequence that specifies the PWA regions over a prediction horizon of $N$.
The control input at state $x$ is then computed via the following optimization problem:
\begin{subequations}\label{eq:mpc_full}
	\begin{align}
		J(x) = &\min_{\mathbf{x}, \mathbf{u}, \bm{\delta}} \: J(\mathbf{x}, \mathbf{u}) \\
		\text{s.t.} \quad &x(0) = x, \label{eq:first_cnstr} \\
		&x(k+1) = A_{\delta(k)} x(k) + B u(k) + c_{\delta(k)} \label{eq:dynamic_constraint} \\
		&\quad\text{for} \: k = 0, \dots, N-1, \nonumber\\
		&x(k) \in \bar{P}_{\delta(k)} \:\:\text{for}\:\: k = 0, \dots, N,\label{eq:region_constraint} \\
		&\big(x(k), u(k)\big) \in \mathcal{X} \times \mathcal{U} \:\:\text{for}\:\: k = 0, \dots, N-1, \\
		&x(N) \in \mathcal{X}_\text{f},\label{eq:second_cnstr}
	\end{align}
\end{subequations}
where $\mathbf{x} = \big(x^\top(0),\dots,x^\top(N)\big)^\top$ and $\mathbf{u} = \big(u^\top(0),\dots,u^\top(N-1)\big)^\top$.
In \eqref{eq:region_constraint} the closure $\bar{P}_{\delta(k)}$ can be used as the dynamics are continuous.
\end{color}
 
If no solution exists for problem \eqref{eq:mpc_full}, by convention $J(x) = \infty$.
Define the set of feasible states $\mathcal{X}_0 = \big\{x | J(x) < \infty\big\}$.
Once \eqref{eq:mpc_full} is solved online numerically, the first of the optimal control inputs $u^\ast(0)$ is applied to the system, with the problem resolved again at each time step in a receding horizon fashion.
However, due to optimizing over the switching sequence $\bm{\delta}$, problem \eqref{eq:mpc_full} is non-convex, and may be computationally intensive to solve.
We hence introduce an alternative MPC controller as a function of both $x$ and $\bm{\delta}$
\begin{equation}\label{eq:mpc}
	J(x, \bm{\delta}) = \min_{\mathbf{x}, \mathbf{u}} J(\mathbf{x}, \mathbf{u}), \quad \text{s.t.} \: \eqref{eq:first_cnstr}-\eqref{eq:second_cnstr},
\end{equation}
where again $J(x, \bm{\delta}) = \infty$ if no solution exists.
With $\bm{\delta}$ prespecified, the dynamics \eqref{eq:dynamic_constraint} are linear time-varying and, for convex $J$, problem \eqref{eq:mpc} is a convex problem that can be solved efficiently online.
Define the set of feasible states for a given $\bm{\delta}$, with $\delta(0)=i$, as $\mathcal{X}_{0, \bm{\delta}} = \big\{x | J(x, \bm{\delta}) < \infty\big\} \subseteq \bar{P}_i \cap \mathcal{X}_0$.
As \eqref{eq:mpc} includes polytopic constraints and linear dynamics, $\mathcal{X}_{0, \bm{\delta}}$ is a polytope \cite{borrelliPredictiveControlLinear2016}.
In the following we address how $\bm{\delta}$ is prespecified, and how it is ensured that $J(x, \bm{\delta}) < \infty$ if $J(x) < \infty$.

\section{STRUCTURE OF $J(x, \vect{\delta})$}\label{sec:properties}
\begin{color}{black}
In this section we analyze the structure of the optimal switching sequence $\bm{\delta}^\ast(x) = \argmin_{\bm{\delta}} J(x, \bm{\delta})$, and provide a result on the feasibility of choices for $\bm{\delta}$, i.e., when $J(x, \bm{\delta}) < \infty$, that is useful in the sequel.
Note that $\bm{\delta}^\ast(x)$ is not necessarily unique due to the closure $\bar{P}_{\delta}$ in \eqref{eq:region_constraint}, e.g., when $x \in \bar{P}_i\cap\bar{P}_j$ then there are at least two valid $\bm{\delta}^\ast(x)$ values, with $i$ and $j$ as the first elements, respectively.

Consider cost functions of the form
\begin{equation}
	J(\mathbf{x}, \mathbf{u}) = \sum_{k=0}^{N-1} \Big(\big\|Qx(k)\big\|_{p_1} + \big\|Ru(k)\big\|_{p_2}\Big) + \big\|Px(N)\big\|_{p_3},
\end{equation}
where $p_1, p_2, p_3\in \{1,\infty\}$, and $Q$, $R$, and $P$ are of rank $n, m$, and $n$, respectively, such that \eqref{eq:mpc_full} can be reformulated as an MILP and \eqref{eq:mpc} as a linear program (LP).
\subsection{Structure of $\vect{\delta}^\ast(x)$}
In \cite{borrelliDynamicProgrammingConstrained2005}, the structure of the optimal control law $u^\ast(0)$ for \eqref{eq:mpc_full} is proven.
Here, we use a similar line of reasoning to show the structure of $\bm{\delta}^\ast(x)$.
A graphical illustration of the following Proposition is given in Figure \ref{fig:prop_1}.
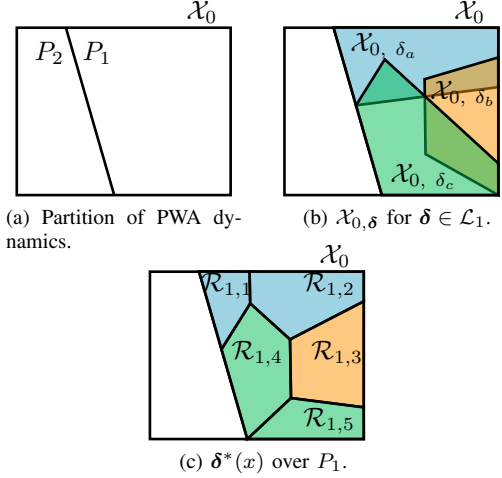
\begin{figure}
	\centering
	\subfloat[Partition of PWA dynamics.\label{fig:prop_1_a}]{\tikzset{every picture/.style={line width=1pt}} 
\begin{tikzpicture}[x=0.75pt,y=0.75pt,yscale=-0.4,xscale=0.4]
	
	\draw   (178,33) -- (446.33,33) -- (446.33,243) -- (178,243) -- cycle ;
	\draw    (239.53,32.67) -- (300.2,243) ;
	
	\draw (240,41.6) node [anchor=north west][inner sep=0.75pt]    {$ \begin{array}{l}
			P_{1}\\
		\end{array}$};
	\draw (180,41.6) node [anchor=north west][inner sep=0.75pt]    {$ \begin{array}{l}
			P_{2}\\
		\end{array}$};
	\draw (370,-10) node [anchor=north west][inner sep=0.75pt]    {$ \begin{array}{l}
			\mathcal{X}_{0}\\
		\end{array}$};

\end{tikzpicture}} \hspace{0.3cm}
	\subfloat[$\mathcal{X}_{0, \bm{\delta}}$ for $\bm{\delta} \in \mathcal{L}_1$.\label{fig:prop_1_b}]{\tikzset{every picture/.style={line width=1pt}}
\begin{tikzpicture}[x=0.75pt,y=0.75pt,yscale=-0.4,xscale=0.4]
	
	\draw   (178,33) -- (446.33,33) -- (446.33,243) -- (178,243) -- cycle ;
	\draw    (239.53,32.67) -- (300.2,243) ;
	\draw  [fill={rgb, 255:red, 62; green, 167; blue, 199 }  ,fill opacity=0.5 ] (446.33,33) -- (446.22,107.57) -- (267.82,130.37) -- (239.53,32.67) -- cycle ;
	\draw  [fill={rgb, 255:red, 252; green, 146; blue, 0 }  ,fill opacity=0.5 ] (446.28,70.28) -- (446.33,243) -- (355.02,191.57) -- (353.82,97.57) -- cycle ;
	\draw  [fill={rgb, 255:red, 3; green, 189; blue, 87 }  ,fill opacity=0.5 ] (303.82,72.77) -- (446.22,203.17) -- (446.33,243) -- (300.2,243) -- (267.82,130.37) -- cycle ;
	
		\draw (370,-10) node [anchor=north west][inner sep=0.75pt]    {$ \begin{array}{l}
			\mathcal{X}_{0}\\
		\end{array}$};
	\draw (240,30) node [anchor=north west][inner sep=0.75pt]    {$ \begin{array}{l}
			\mathcal{X}_{0,\ \delta_{a}}\\
		\end{array}$};
	\draw (340,90) node [anchor=north west][inner sep=0.75pt]    {$ \begin{array}{l}
			\mathcal{X}_{0,\ \delta_{b}}\\
		\end{array}$};
	\draw (290,195) node [anchor=north west][inner sep=0.75pt]    {$ \begin{array}{l}
			\mathcal{X}_{0,\ \delta_{c}}\\
		\end{array}$};

\end{tikzpicture}}\\
	\vspace{-0.6cm}
	\subfloat[$\bm{\delta}^\ast(x)$ over $P_1$.\label{fig:prop_1_c}]{\tikzset{every picture/.style={line width=1pt}} 

\begin{tikzpicture}[x=0.75pt,y=0.75pt,yscale=-0.4,xscale=0.4]
	
	\draw   (178,33) -- (446.33,33) -- (446.33,243) -- (178,243) -- cycle ;
	\draw    (239.53,32.67) -- (300.2,243) ;
	\draw  [fill={rgb, 255:red, 3; green, 189; blue, 87 }  ,fill opacity=0.5 ] (303.82,72.77) -- (353.82,118.01) -- (355.02,191.57) -- (300.2,243) -- (267.82,130.37) -- cycle ;
	\draw  [fill={rgb, 255:red, 62; green, 167; blue, 199 }  ,fill opacity=0.5 ] (303.02,32.81) -- (303.82,72.77) -- (267.82,130.37) -- (239.53,32.67) -- cycle ;
	\draw  [fill={rgb, 255:red, 62; green, 167; blue, 199 }  ,fill opacity=0.5 ] (303.02,32.81) -- (303.82,72.77) -- (353.82,118.01) -- (446.28,70.28) -- (446.33,33) -- cycle ;
	\draw  [fill={rgb, 255:red, 252; green, 146; blue, 0 }  ,fill opacity=0.5 ] (353.82,118.01) -- (446.28,70.28) -- (446.22,203.17) -- (355.02,191.57) -- cycle ;
	\draw  [fill={rgb, 255:red, 3; green, 189; blue, 87 }  ,fill opacity=0.5 ] (446.22,203.17) -- (446.33,243) -- (300.2,243) -- (355.02,191.57) -- cycle ;
	
		\draw (370,-10) node [anchor=north west][inner sep=0.75pt]    {$ \begin{array}{l}
			\mathcal{X}_{0}\\
		\end{array}$};
	\draw (222,25) node [anchor=north west][inner sep=0.75pt]    {$ \begin{array}{l}
			\mathcal{R}_{1,1}\\
		\end{array}$};
	\draw (350,25) node [anchor=north west][inner sep=0.75pt]    {$ \begin{array}{l}
			\mathcal{R}_{1,2}\\
		\end{array}$};
	\draw (360,110) node [anchor=north west][inner sep=0.75pt]    {$ \begin{array}{l}
			\mathcal{R}_{1,3}\\
		\end{array}$};
	\draw (350,195) node [anchor=north west][inner sep=0.75pt]    {$ \begin{array}{l}
			\mathcal{R}_{1,5}\\
		\end{array}$};
	\draw (260,110) node [anchor=north west][inner sep=0.75pt]    {$ \begin{array}{l}
			\mathcal{R}_{1,4}\\
		\end{array}$};

\end{tikzpicture}}
	\caption{Illustration of the proof of Proposition \ref{prop:1}. $\mathcal{R}_{1, 1}$ and $\mathcal{R}_{1, 2}$ are associated with $\bm{\delta}^\ast(x) = \bm{\delta}_a$, $\mathcal{R}_{1, 3}$ is associated with $\bm{\delta}^\ast(x) = \bm{\delta}_b$, and $\mathcal{R}_{1, 4}$ and $\mathcal{R}_{1, 5}$ are associated with $\bm{\delta}^\ast(x) = \bm{\delta}_c$.}
	\label{fig:prop_1}
\end{figure}
\begin{proposition}\label{prop:1}
	Consider the problem $\bm{\delta}^\ast(x) \in \argmin_{\bm{\delta}} J(x, \bm{\delta})$ for $x \in \mathcal{X}_0$.
	Define $\mathcal{L}_i = \Big\{\bm{\delta} \in \mathcal{Z}_{[0, l]}^{N+1} \big| \delta(0) = i\Big\}$. Then
	\begin{enumerate}
		\item $\bm{\delta}^\ast(x) \in \mathcal{L}_i$ for all $x \in \bar{P}_i$.
		\item For $i \in \mathcal{L}$ we can define the collection $\mathcal{R}_{i, 1}, \dots, \mathcal{R}_{i, R_i}$ as a convex partition of $P_i \cap \mathcal{X}_0$, where $\bm{\delta}^\ast(x) = \bm{\delta}^\ast_{i, j}$ for all $x \in \bar{\mathcal{R}}_{i, j}$.
		\item The collection $\mathcal{R}_{1,1},\dots,\mathcal{R}_{1, R_1},\dots,\mathcal{R}_{l,1},\dots,\mathcal{R}_{l, R_l}$ is a convex partition of $\mathcal{X}_0$.
	\end{enumerate}
\end{proposition}
\begin{proof}
	\textit{1)} For $x \in \bar{P}_i$ there exists a $\bm{\delta}^\ast(x)$ for which the first element must be $i$ by constraints \eqref{eq:first_cnstr} and \eqref{eq:region_constraint} for $k=0$.
	Hence, $\bm{\delta}^\ast(x) \in \mathcal{L}_i$ for all $x \in \bar{P}_i$.
	\textit{2)} As states can be feasible for multiple choices of $\bm{\delta}$, the sets $\mathcal{X}_{0, \bm{\delta}}$ can overlap, e.g., Figure \ref{fig:prop_1_b}.
	An optimal $\bm{\delta}^\ast(x)$ for a state $x \in \bar{P}_i$ where multiple switching sequences $\bm{\delta}_a, \bm{\delta}_b, \dots \in \mathcal{L}_i$ are feasible is found by comparing the values $J(x, \bm{\delta}_a), J(x, \bm{\delta}_b), \dots,$ and taking the minimum.
	Consider the simple case of two such sets $\mathcal{X}_{0, \bm{\delta}_a} \cup \mathcal{X}_{0, \bm{\delta}_b} \neq \emptyset$.
	Then $\bm{\delta}^\ast(x) = \bm{\delta}_a$ for $x$ in
	\begin{equation}
		\big\{x | \mathcal{X}_{0, \bm{\delta}_a} \cup \mathcal{X}_{0, \bm{\delta}_b}, J(x, \bm{\delta}_a) - J(x, \bm{\delta}_b) \leq 0\big\} \cap (\mathcal{X}_{0, \bm{\delta}_a} \setminus \mathcal{X}_{0, \bm{\delta}_b}),
	\end{equation}
	which, as $J(x, \bm{\delta}_a)$ and $J(x, \bm{\delta}_b)$ are piecewise-linear functions of $x$ \cite{borrelliPredictiveControlLinear2016}, can be represented as the union of a finite number of polytopes.
	Generalizing this to many overlapping sets, the region defining where a given sequence is optimal can be decomposed into a finite number of polytopes, i.e., the regions $\mathcal{R}_{i, j}$.
	with the definition of $\mathcal{L}_i$ we have $\bigcup_{\bm{\delta} \in \mathcal{L}_i} \mathcal{X}_{0, \bm{\delta}} = \bar{P}_i \cap \mathcal{X}_0$, e.g., Figure \ref{fig:prop_1_c}.
	It follows that the collection $\mathcal{R}_{i,1},\dots,\mathcal{R}_{i, R_i}$ is a convex partition of $P_i \cap \mathcal{X}_0$.
	\textit{3)} Finally, as $\bigcup_{i \in \mathcal{L}} (P_i \cap \mathcal{X}_0) = \mathcal{X}_0$, the collection $\mathcal{R}_{1,1},\dots,\mathcal{R}_{1, R_1},\dots,\mathcal{R}_{l,1},\dots,\mathcal{R}_{l, R_l}$ is a convex partition of $\mathcal{X}_0$.
\end{proof}
This result informs us that $\bm{\delta}^\ast(x)$ can be characterized by a convex partition, where each polytope in the partition is associated with a value for $\bm{\delta}^\ast(x)$.
Furthermore, as $\mathcal{X}_0$ can be expressed as a union of polytopes, $\mathcal{X}\setminus\mathcal{X}_0$ can also be expressed as a union of polytopes, such that by learning a convex partition over $\mathcal{X}$ we can represent both $\mathcal{X}_0$ and $\bm{\delta}^\ast(x)$ over $\mathcal{X}_0$.\end{color}

\subsection{Feasibility of $J(x, \vect{\delta})$}
\begin{proposition}\label{prop:2}
	For any $\bm{\delta}$ and the set of states $\tilde{\mathcal{X}} \subseteq \mathcal{X}_0$ such that $J(x, \bm{\delta}) < \infty$ for all $x \in \tilde{\mathcal{X}}$, we have $J(x, \bm{\delta}) < \infty$ for all $x \in \text{Conv}(\tilde{\mathcal{X}})$.
\end{proposition}
\begin{proof}
	Assume there exists $x \in \text{Conv}(\tilde{\mathcal{X}})$ such that $J(x, \bm{\delta}) = \infty$. 
	Then the set $\mathcal{X}_{0, \bm{\delta}}$ is non-convex.
	This is a contradiction; hence, $J(x, \bm{\delta}) < \infty$ for all $x \in \text{Conv}(\tilde{\mathcal{X}})$.
\end{proof}
Proposition \ref{prop:2} provides a tool for checking the feasibility of a given $\bm{\delta}$ over an infinite number of states by checking feasibility over a finite number of states, i.e., given a polytope $\mathcal{R} \subseteq \mathcal{X}_0$, the feasibility of $\bm{\delta}$ for the \emph{entire} polytope can be confirmed by confirming that $J(x, \bm{\delta}) < \infty$ for all $x \in \text{Vert}(\mathcal{R})$.
\section{LEARNING-BASED MPC CONTROLLER}\label{sec:learning_based_mpc}
In this section we formulate our learning-based MPC solution.
We propose to learn a policy $\pi_\theta(x) = \bm{\delta}$ offline in a supervised-learning manner, where $\theta$ is a vector of parameters to be learned, and pairs of training data $\big(x, \bm{\delta}^\ast(x)\big)$ are generated by solving \eqref{eq:mpc_full}.
Then, for deployment online, only the LP \eqref{eq:mpc} must be solved for $J\big(x, \pi_\theta(x)\big)$.

\subsection{Classifier}
Proposition \ref{prop:1} suggests that $\pi_\theta$ should be a classifier that can effectively approximate a convex partition, while Proposition \ref{prop:2} provides a tool to verify feasibility, provided $\pi_\theta$ partitions its outputs into polytopes.
In light of this, we formulate the following structure that defines a class of suitable classifiers: $\pi_\theta(x) = \phi_{\theta, i}\big(f_{\theta, i}(x)\big)$ for $x \in P_i$, with
\begin{equation} \label{eq:classifier_struct}
	\begin{aligned}
		&f_{\theta, i}: P_i \to \{\mathcal{R}_{i, 1}, \dots, \mathcal{R}_{i, \bar{R}_i}\}\\
		&\phi_{\theta, i}: \{\mathcal{R}_{i, 1}, \dots, \mathcal{R}_{i, \bar{R}_i}\} \to \mathcal{L}_i \cup \{-1\},
	\end{aligned}
\end{equation}
where each collection $\mathcal{R}_{i, 1},\dots,\mathcal{R}_{i, \bar{R}_i}$ is a convex partition of $P_i$.
The function $\phi_{\theta, i}$ maps the polytopes to either a switching sequence $\bm{\delta}$, or to $-1$ representing infeasibility, i.e., $x \notin \mathcal{X}_0$.
For convenience, denote the codomains of $f_{\theta, i}$ and $\phi_{\theta, i}$ as $\mathcal{F}_{\theta, i}$ and $\Phi_i$, respectively.
Note that the number of polytopes $\bar{R}_i$ depends on the training data.

Many classifiers can satisfy this structure, e.g., decision trees, oblique decision trees, neural networks with linear activation functions, and support vector machines \cite{bishop2006pattern}.
In \cite{bemporadPiecewiseLinearRegression2023} a classifier tailored to learning functions that are PWA across polytopes is presented, alongside a thorough review of alternative approaches.

\subsection{Iterative Training Procedure}
Learning $\pi_\theta$ in a supervised manner requires the solution to problem \eqref{eq:mpc_full} in order to generate training data.
As solving \eqref{eq:mpc_full} is computationally demanding, it is desirable to train $\pi_\theta$ with as little data as possible.
Furthermore, given the structure of $\bm{\delta}^\ast(x)$ outlined in Section \ref{sec:properties}, it is clear that some states are more informative than others.
In particular, states near the boundaries of the polytopes forming the convex partition are more useful than those in the interiors. 

\begin{algorithm}
	\caption{Train $\pi_\theta$ (offline)}\label{alg:train}
	\begin{algorithmic}[1]
		\State \textbf{Inputs}: Initial training set $\mathcal{T} \subset \mathcal{X} \times \bigcup_{i \in \mathcal{L}} \Phi_i$
		\While{true}
			\State $\theta \gets \text{train}(\mathcal{T})$ \label{step:train}
			\State $\text{feasible} \gets 1$
			\State $\mathcal{Z} \gets \{\}$
			\For{$i \in \mathcal{L}$}
				\For{$\mathcal{R} \in \mathcal{F}_{\theta, i}$}
					\For{$x \in \text{Vert}(\mathcal{R})$}
						\If{$\Big(\phi_{\theta, i}(\mathcal{R}) = -1$ and $J_\text{t}(x) < \infty\Big)$ or $\Big(\phi_{\theta, i}(\mathcal{R}) \neq -1$ and $J\big(x, \pi_\theta(x)\big) = \infty\Big)$}\label{step:big_if}
							\State $\mathcal{Z} \gets \mathcal{Z} \cup \{x\}$
							\State feasible $\gets 0$
						\EndIf
					\EndFor
				\EndFor
			\EndFor
			\If{feasible $= 1$}
				\State \textbf{Return}: $\pi_\theta$
			\Else
				\For{$x \in \mathcal{Z}$}
					\State Solve \eqref{eq:mpc_full} for $J(x)$ and $\bm{\delta}^\ast$\label{step:solve}
					\If{$J(x) < \infty$}
						\State $\mathcal{T} \gets \mathcal{T} \cup \{(x, \bm{\delta}^\ast)\}$
					\Else
						\State $\mathcal{T} \gets \mathcal{T} \cup \{(x, -1)\}$
					\EndIf
				\EndFor
			\EndIf
		\EndWhile
	\end{algorithmic}
\end{algorithm}
Addressing both these points, Algorithm \ref{alg:train} details the offline training procedure.
In step \ref{step:train}, train($\mathcal{T}$) is pseudo-code for a supervised learning algorithm that trains the classifier given a set of labeled data $\mathcal{T} = \big\{\big(x^{(i)}, \bm{\delta}^{\ast}(x^{(i)})\big)\big\}_{i}^{N_\text{data}}$, and is specific to the form of $\pi_\theta$.
Also specific to the form of $\pi_\theta$ is the presence and degree of classifier misclassification, i.e., how often $\pi_\theta$ chooses the incorrect switching sequence for a state in the training data $\mathcal{T}$,
which can aid in finding a feasible $\pi_\theta$ with a simple partition, i.e., with less polytopes, but can increase suboptimality. 
Beginning from an initial training set, $\pi_\theta$ is trained iteratively, where at each iteration, all vertices of all learned polytopes are checked for feasibility.
Only vertices that are infeasible are added to the training set, naturally concentrating the training data in the areas of high information gain, i.e., at the boundaries of the polytopes of the partition.
\begin{remark}
	Clearly the number of iterations and the amount of data required for convergence of Algorithm \ref{alg:train} depends on the underlying structure of $\bm{\delta}^\ast(x)$, which in general grows in complexity with $N$ and $l$.
	An analysis of the scalability of Algorithm \ref{alg:train} with $N$ and $l$ is left for future work.
\end{remark}

\subsection{Approximating $\mathcal{X}_0$}
\begin{color}{black}
In order to approximate the feasible region $\mathcal{X}_0$, the boundary of $\mathcal{X}_0$ is handled as a special case.
As the PWA dynamics are continuous, several switching sequences can be feasible for a given state, e.g., $x \in \bar{P}_i\cap\bar{P}_j$, such that the shared vertices of adjacent polytopes can be feasible for the switching sequence of each polytope.
However, it is not the case that the shared vertices of polytopes can be used to verify feasibility and infeasibility, as this would imply both $J(x) < \infty$ and $J(x) = \infty$ at the boundary of $\mathcal{X}_0$.
In light of this, we introduce the tightened version of \eqref{eq:mpc_full}:
\begin{equation}\label{eq:mpc_tight}
	\begin{aligned}
		J_\text{t}(x) = &\min_{\mathbf{x}, \mathbf{u}, \bm{\delta}} \: J(\mathbf{x}, \mathbf{u}) \\
		\text{s.t.} \quad &\eqref{eq:first_cnstr}-\eqref{eq:region_constraint} \\
		&\big(x(k), u(k)\big) \in \big(\mathcal{X} \ominus \mathcal{O}\big) \times \mathcal{U} \nonumber \\
		&\quad \text{for}\:\: k = 0, \dots, N-1, \\
		&x(N) \in \mathcal{X}_\text{f}\ominus \mathcal{O},
	\end{aligned}
\end{equation}
where $\mathcal{O}$ is an arbitrarily small set containing the origin.
Again, if no solution exists for \eqref{eq:mpc_tight} then $J_\text{t}(x) = \infty$.
The feasible set for \eqref{eq:mpc_tight} is contained within $\mathcal{X}_0$, such that on the boundary of $\mathcal{X}_0$ we have $J(x) < \infty$ and $J_\text{t}(x) = \infty$.
In Algorithm \ref{alg:train}, when verifying that the vertices of a polytope labeled with $-1$ are infeasible, in step \ref{step:big_if}, we then use $J_\text{t}(x)$.
The result is that an under-approximation of $\mathcal{X}_0$ is learned.
However the approximation error can be made arbitrarily small by reducing $\mathcal{O}$.  
Define this feasible set, representing the valid region for $\pi_\theta$, as $\mathcal{X}_{0, \theta} = \big\{x | \pi_\theta(x) \neq -1\big\}$.\end{color}
\begin{proposition}\label{prop:3}
	For $\pi_\theta$ trained with Algorithm \ref{alg:train}, $J\big(x, \pi_\theta(x)\big) < 0$ for all $x \in \mathcal{X}_{0, \theta}$.
\end{proposition}
\begin{proof}
	For $\pi_\theta$ trained with Algorithm \ref{alg:train} it holds, for all $i \in \mathcal{L}$ and for all $\mathcal{R} \in \mathcal{F}_{\theta, i}$ such that $\phi(\mathcal{R}) \neq -1$, that $J\big(x, \pi_\theta(x)\big) < \infty$ for all $x \in \text{Vert}(\mathcal{R})$ and hence, by Proposition \eqref{prop:2}, for all $x \in \bar{\mathcal{R}}$.
	As, by the definition \eqref{eq:classifier_struct}, the collection of sets $\mathcal{R}$ for all $\phi_{\theta, i}(\mathcal{R}) \neq -1$ is a convex partition of $P_i \cap \mathcal{X}_{0, \theta}$, it holds that $J\big(x, \pi_\theta(x)\big) < 0$ for all $x \in \mathcal{X}_{0, \theta}$.
\end{proof}
\begin{remark}
	In Algorithm \ref{alg:train}, an MILP is solved only when generating training data at the infeasible vertices, i.e., step \ref{step:solve}.
	Checking for $J_\text{t}(x) < \infty$ in step \ref{step:big_if} involves checking for the existence of a feasible solution only, which is much easier than solving the problem to optimality.
\end{remark}

\subsection{Control Law}
With $\pi_\theta$ we now define the controller in Algorithm \ref{alg:control}.
The check in step \ref{step:check} is required as $\mathcal{X}_{0, \theta}$ can be an under-approximation of $\mathcal{X}_0$, such that a feasible trajectory could in theory pass out of $\mathcal{X}_{0, \theta}$ before entering $\mathcal{X}_\text{f}$.
We now prove recursive feasibility of the controller.
\begin{proposition}\label{prop:4}
	For $x_0 \in \mathcal{X}_{0, \theta}$, for the closed-loop system and with $\bm{\delta}_k$ chosen as by Algorithm \ref{alg:control}, we have $J(x_k, \bm{\delta}_k) < \infty$ for all $k \geq 0$.
\end{proposition}
\begin{proof}
	Assume that at time step $k$ we have $J(x_k, \bm{\delta}_k) < \infty$. 
	Applying $u_k^\ast(0)$, obtained from \eqref{eq:mpc}, will propagate the system to $x_{k+1} = x_k^\ast(1)$.
	If $\pi_\theta(x_{k+1}) = -1$, then the shifted solution $\bar{u}_k = \big(u_k^\ast(1),\dots,u_k^\ast(N-1), K_jx_k^\ast(N)\big)$ is feasible for problem \eqref{eq:mpc} with the shifted sequence $\bm{\delta}_{k+1} = \big(\delta_{k}(1),\dots,\delta_{k}(N), i\big)$ and with $(A_j + B K_j)x_{k}^\ast(N) \in P_i$ and $x_{k}^\ast(N) \in P_j$, as $\mathcal{X}_\text{f}$ is forward invariant under the linear controllers.
	Hence, we have $J(x_{k+1}, \bm{\delta}_{k+1}) < \infty$.
	If $\pi_\theta(x_{k+1}) \neq -1$, then $\bm{\delta}_{k+1} = \pi_\theta(x_{k+1})$ and $J(x_{k+1}, \bm{\delta}_{k+1}) < \infty$ by Proposition \ref{prop:3}.
	Therefore, if $J(x_k, \bm{\delta}_k) < \infty$ at time step $k$ then we have $J(x_{k+1}, \bm{\delta}_{k+1}) < \infty$ at time step $k+1$.
	For $x_0 \in \mathcal{X}_{0, \theta}$ we have $J(x_{0}, \bm{\delta}_{0}) < \infty$ again by Proposition \ref{prop:3}, and by induction  $J(x_{k}, \bm{\delta}_{k}) < \infty$ for all $k \geq 0$.
\end{proof}
\begin{algorithm}
	\caption{Controller at time step $k$ (online)}\label{alg:control}
	\begin{algorithmic}[1]
		\State \textbf{Inputs}: Current state $x_k$, previous switching sequence $\bm{\delta}_{k-1}$ (except if $k = 0$)
		\If{$k = 0$ or $\pi_\theta(x_k) \neq -1$}\label{step:check}
			\State $\bm{\delta}_k \gets \pi_\theta(x_k)$
		\Else
			\State $\bm{\delta}_k \gets \big(\delta_{k-1}(1),\dots,\delta_{k-1}(N), i\big)$ with $(A_j + B K_j)x_{k-1}^\ast(N) \in P_i$ and $x_{k-1}^\ast(N) \in P_j$
		\EndIf
		\State Solve $J(x_k, \bm{\delta}_k)$ and apply $u^\ast(0)$		
	\end{algorithmic}
\end{algorithm}

\section{NUMERICAL EXAMPLE}\label{sec:example}
Consider a representative numerical PWA system \eqref{eq:dynamics} with $l = 2$, 
\begin{equation}
	A_1 = \begin{bmatrix}
		1 & 0.2 \\ 0 & 1
	\end{bmatrix} \: A_2 = \begin{bmatrix}
	0.5 & 0.2 \\ 0 & 1
	\end{bmatrix} \: c_1 = \begin{bmatrix}
	0 \\ 0
	\end{bmatrix} \: c_2 = \begin{bmatrix}
	0.5 \\ 0
	\end{bmatrix},
\end{equation}
$B = \begin{bmatrix}
	0.1 & 1
\end{bmatrix}^\top$, and $P_{1(2)} = \{x | \begin{bmatrix}
1 & 0
\end{bmatrix}x \leq(>) 1\}$.
The input and state constraints are $\mathcal{U} = \big\{u \big| |u| \leq 3\big\}$ and $\mathcal{X} = \{x | Dx \leq E\}$, respectively, with
\begin{equation}
	\begin{aligned}
		D &= \begin{bmatrix}
			-1 & -3 & 0.2 & -1 & 1 & 0 \\
			1 & -1 & 1 & 0 & 0 & -1
		\end{bmatrix}^\top \\ E &= \begin{bmatrix}
			15 & 25 & 9 & 6 & 8 & 10
		\end{bmatrix}^\top.
	\end{aligned}
\end{equation}
Consider the MPC controller \eqref{eq:mpc_full} with cost matrices $Q = P = \mathbf{I}_2$, $R = 1$, and using the 1-norm, i.e., $p_1 = p_2 = p_3 = 1$.
The terminal set is computed as the maximal constraint admissible set for the system $x(k+1) = (A_1 + B K)x$, restricted to $P_1$, where $K$ is the optimal LQR gain controller for $(A_1, B, Q, R)$.
The example is simulated on an 11th
Gen Intel laptop with four i7 cores, 3.00GHz clock speed,
and 16Gb of RAM.
For fairness, all MILP and LP optimization problems are solved with Gurobi \cite{gurobi} without warm-starting, while explicit MPC controllers are generated with the MPT3 toolbox \cite{MPT3}.
Source code is available at \url{https://github.com/SamuelMallick/supervised-learning-pwa-mpc}.
For $\pi_\theta$ we use an ensemble of the classifier presented in \cite{bemporadPiecewiseLinearRegression2023}, with one for each $P_i$, with all tunable hyperparameters for the classifier available in the source code.

Algorithm \ref{alg:train} is used to train $\pi_\theta$ using an initial training set $\mathcal{T} = \big\{\big(x^{(i)}, \bm{\delta}^{\ast}(x^{(i)})\big)\big\}_{i = 1}^{45} \cup \big\{\big(x^{(j)}, \bm{\delta}^{\ast}(x^{(j)})\big)\big\}_{j = 1}^{45}$ with $x^{(i)}$ and $x^{(j)}$ sampled uniformly from $P_1$ and $P_2$, respectively.
The tightened MPC problem \eqref{eq:mpc_tight} is implemented with $\mathcal{O}=\big\{x \big| \|x\|_\infty < 0.1\big\}$.
Figure \ref{fig:training} demonstrates Algorithm \ref{alg:train} at iterations 0, 8 and 42, for $N=12$.
It can be seen that across the iterations, training data is added around areas of importance, allowing a feasible $\pi_\theta$ to be learned efficiently.
At the first iteration training data is sampled randomly, and the partition is poor, with many infeasible vertices.
By the eighth iteration, following Algorithm \ref{alg:train}, more points have been sampled at the infeasible vertices of the previous iterations, and the partition is significantly improved.
Finally, at iteration 42 the partition is validated to be feasible for all $x \in \mathcal{X}_{0, \theta}$, and the training is terminated. 
\begin{figure}
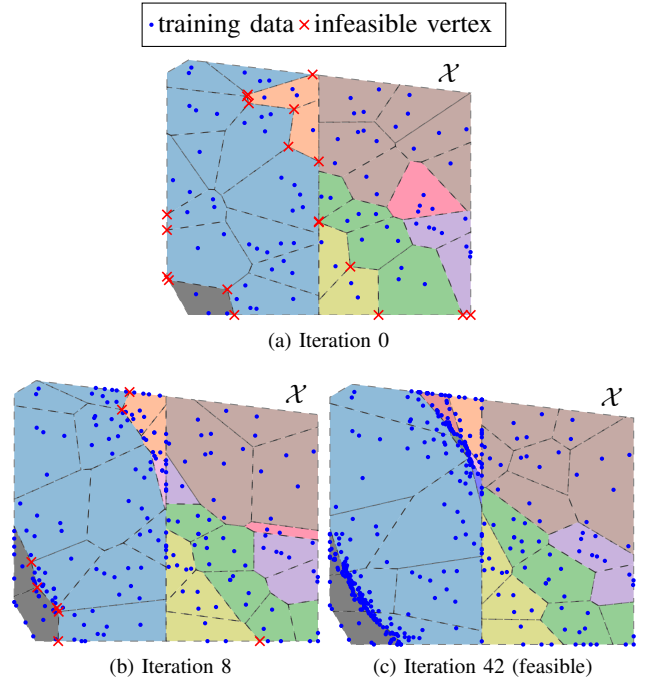

\centering
\subfloat[Iteration 0]{\input{media/tikz/figure_1_parc_0}} \\
\subfloat[Iteration 8]{\input{media/tikz/figure_1_parc_8}}
\subfloat[Iteration 42 (feasible) \label{fig:training_c}]{\input{media/tikz/figure_1_parc_42}}
\caption{Training for $N=12$. Each color represents a different $\bm{\delta}$, while the dark gray region represents infeasibility.}
\label{fig:training}
\end{figure}

Table \ref{tab:our} shows the number of regions in the learned convex partition of $\mathcal{X}$ and the number of iterations required to guarantee feasibility, as the horizon $N$ increases from 5 to 12.
Additionally, the number of training data generated is reported, normalized by that for $N=5$, indicating how the offline computation scales with $N$.
For comparison, an explicit MPC controller \cite{oberdieckExplicitHybridModelPredictive2015} is computed for each $N$.
This involves solving an mp-MILP, resulting in a partition of $\mathcal{X}$ that defines the optimal control law.
Table \ref{tab:exp} shows the number of regions required by the explicit MPC approach, and the computation time required to compute the controller, again normalized by the time for $N=5$.
For $\pi_\theta$, both the number of regions and offline computation are relatively uncorrelated with $N$, while for explicit MPC the number of regions and the offline computation time scale poorly with $N$.
\begin{table}
	\centering
	\caption{Training $\pi_\theta$: Number of training data normalized by the value for $N=5$, with true values in brackets.}
	\begin{tabular}{c||c|c|c}
		\hline
		N & \# regions & \# iterations &\# training data \\ \hline
		5 & 29 & 34& 1 \hfill (377) \\ \hline
		6 & 19 & 46 & 1.170 \hfill (441) \\ \hline
		7 & 19 & 56 & 1.191\hfill (449) \\ \hline
		8 & 26 & 26 & 0.801\hfill (302) \\ \hline
		9 & 27 & 33 & 0.889\hfill (335) \\ \hline
		10 & 24 & 52 & 1.016\hfill (383) \\ \hline
		11 & 25 & 42 & 0.939\hfill (354) \\ \hline
		12 & 30 & 42 & 1.053\hfill (397) \\ \hline
	\end{tabular}
	\label{tab:our}
\end{table}
\begin{table}
	\centering
	\caption{Explicit MPC: Computation time normalized by the value for $N=5$, with true values in brackets.}
	\begin{tabular}{c||c|c}
		\hline
		N & \# regions & comp. time (s) \\ \hline
		5 & 388 & 1\hfill (25.189)  \\ \hline
		6 & 469 & 5.908\hfill (148.807)  \\ \hline
		7 & 520 & 14.781\hfill (372.299)  \\ \hline
		8 & 549 & 40.761\hfill (1.027$\cdot10^3$)  \\ \hline
		9 & 574 & 150.494\hfill (3.791$\cdot10^3$)  \\ \hline
		10 & 600 &464.185\hfill (1.169$\cdot10^4$)  \\ \hline
		11 & 741 & 1.307$\cdot10^3$\hfill (3.292$\cdot10^4$)   \\ \hline
		12 & 2541 & 6.965$\cdot10^3$\hfill (1.754$\cdot10^5$)  \\ \hline
	\end{tabular}
	\label{tab:exp}
\end{table}

To investigate the performance of $\pi_\theta$, the open-loop cost $J\big(x, \pi_\theta(x)\big)$ is compared against the optimal $J(x)$ with $N=12$.
Figure \ref{fig:open_loop} shows the percentage suboptimality: $\Delta J = 100 \cdot \Big(J\big(x, \pi_\theta(x)\big) - J(x)\Big) \big/ J(x)$, for 100709 states sampled densely from $\mathcal{X}_{0, \theta}$.
Suboptimality is only introduced around the boundaries of the convex partition in Figure \ref{fig:training_c}, demonstrating that $\pi_\theta$ learns a convex partition close to the optimal one.
Furthermore, for all states $J\big(x, \pi_\theta(x)\big) < \infty$, demonstrating Proposition \ref{prop:3}.

To investigate the closed-loop performance, Algorithm \ref{alg:control} is compared against the optimal online MPC controller\footnote{The performance of the optimal online MPC controller is the same as that of the explicit MPC controller; however, as the memory requirements for explicit MPC can be unreasonable, we use an online approach and explore its online computational burden.}, where \eqref{eq:mpc_full} is solved, for $N = 12$.
Simulations are run for 1000 initial states, sampled uniformly from $\mathcal{X}_{0, \theta}$, until $\|x\|_2 < 0.01$, with the cumulative closed-loop costs compared.
Table \ref{tab:closed_loop} gives the distribution of the percentage suboptimality, as well as the computation time for each approach.
It can be seen that, as the proposed approach optimizes for the continuous variables online, and as $\bm{\delta} = \pi_\theta(x)$ is suboptimal over only a subset of the state space (see Figure \ref{fig:open_loop}), the closed-loop suboptimality is limited.
Furthermore, the computation time is significantly improved by Algorithm \ref{alg:control}.
Finally, for all time steps of all simulations, $J\big(x_k, \bm{\delta}_k\big) < \infty$, demonstrating Proposition \ref{prop:4}.
\begin{figure}
	\centering
	\includegraphics[width=0.3\textheight]{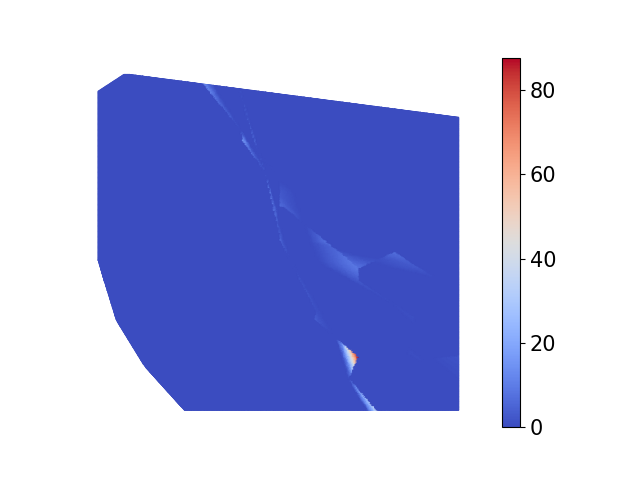}
	\caption{Open-loop suboptimality $\Delta J$ for $N=12$.}
	\label{fig:open_loop}
\end{figure}
\begin{table}
	\centering
	\footnotesize
	\caption{Closed-loop comparison $N = 12$.}
	\begin{tabular}{c|ccccc}
		\hline
		 & mean & median & min & max & std \\
		$\sum_k \Delta J$ & 0.371 & 0 & 0 & 35.970 & 1.808 \\
		Time \eqref{eq:mpc_full} (s) & 0.0323 & 0.0152 & 0.00361 & 0.315 & 0.0362 \\
		Time \eqref{eq:mpc} (s) & 0.00143 & 0.00125 & 0.000605 & 0.0517& 0.00231
	\end{tabular}
	\label{tab:closed_loop}
\end{table}

\section{CONCLUSIONS}\label{sec:conclusions}
In this work we have proposed a learning-based MPC controller for PWA systems where the PWA switching sequences are selected by a learned policy, such that only a linear optimization problem must be solved online.
In contrast to existing works in this direction, feasibility of the controller is guaranteed.  
To this end a class of suitable classifiers has been formulated, and an efficient training strategy has been presented.
A numerical example has demonstrated the benefits of the proposed approach in terms of offline computation and online memory with respect to offline MPC, and online computation with respect to online MPC.
Future work will look at bounding the suboptimality of the proposed approach, as well as further experimentation on larger systems with more regions in the PWA dynamics.




\bibliography{references}

\end{document}